\documentclass[letterpaper]{article}
\newif\ifarxivversion
\arxivversiontrue
\ifarxivversion
  \usepackage[preprint]{aaai2027}
\else
  \usepackage[submission]{aaai2027}
\fi
\usepackage[hyphens]{url}
\usepackage{graphicx}
\usepackage{natbib}
\usepackage{booktabs}
\usepackage{amsmath}
\usepackage{amssymb}
\usepackage{amsthm}
\usepackage{algorithm}
\usepackage{algorithmic}

\newtheorem{theorem}{Theorem}
\newtheorem{lemma}{Lemma}
\newtheorem{proposition}{Proposition}
\newtheorem{corollary}{Corollary}
\newtheorem{assumption}{Assumption}
\theoremstyle{definition}
\newtheorem{definition}{Definition}
\theoremstyle{remark}
\newtheorem{remark}{Remark}

\newcommand{\E}{\mathbb{E}}
\newcommand{\Prb}{\mathbb{P}}
\newcommand{\Reg}{\mathrm{Reg}}
\newcommand{\1}{\mathbf{1}}
\newcommand{\C}{\mathcal C}
\newcommand{\empt}{\emptyset}
\newcommand{\wt}{\widetilde}
\newcommand{\wh}{\widehat}
\newcommand{\UCB}{\mathrm{UCB}}
\newcommand{\FUCB}{\mathrm{F\text{-}UCB}}
\newcommand{\OPT}{\mathrm{OPT}}

\newif\ifcomments
\commentsfalse

\ifcomments
  \newcommand\shi[1]{\textcolor{green}{[Shi: #1]}}
  \newcommand\yc[1]{\textcolor{blue}{[Yiling: #1]}}
    \newcommand\sadie[1]{\textcolor{orange}{[Sadie: #1]}}
\else
  \newcommand\shi[1]{}
  \newcommand\yc[1]{}
  \newcommand\sadie[1]{}
\fi

\title{Contextual Procurement Auctions with Bandit Learning}
\ifarxivversion
\author{Yiling Chen\thanks{Authors are listed alphabetically by last name.},
Shi Feng, and Sadie Zhao}
\affiliations{
Harvard University\\
\texttt{yiling@seas.harvard.edu},
\texttt{shifeng-thu@outlook.com},
\texttt{sadie\_zhao@g.harvard.edu}
}
\else
\author{}
\affiliations{}
\fi

\begin{document}

\maketitle

\begin{abstract}
We study repeated procurement auctions in which producers have private costs and the platform must learn the context-dependent value of selecting each producer. 
We evaluate performance by welfare regret: the cumulative loss in total surplus relative to the full-information efficient rule that knows the context-dependent values and true producer costs.
%
%
The natural UCB allocation rule achieves $\widetilde O(\sqrt{ngT})$ welfare regret under truthful bids, but its adaptive, bid-dependent learning path does not by itself ensure truthfulness. To obtain exact incentives, we first design a bid-independent explore-then-commit mechanism with empirical threshold payments; it is dominant-strategy truthful and has $\widetilde O((ng)^{1/3}T^{2/3})$ regret. We then introduce frozen-payment UCB, which estimates payments from initial bid-independent exploration but continues allocation learning by UCB. Under a truthful-path margin condition, the frozen-payment UCB is approximately truthful with an average per-round deviation gain $\widetilde O(T^{-1/4})$ for fixed $n$, $g$. Under truthful bidding, it achieves $\widetilde O(\sqrt{ngT})$ welfare regret, matching the UCB rate. \yc{I changed the previous sentence. It seems that the regret is still under truthful bidding. There is a bit of subtlety here. }
A lower bound shows that this regret-incentive tradeoff is tight within the frozen critical-payment class.
\end{abstract}

\section{Introduction}

Auctions often allocate scarce resources whose productive value is initially
uncertain, not merely extract revenue~\cite{vickrey1961counterspeculation};
see also Milgrom~\cite{milgrom2004putting} and
Krishna~\cite{krishna2009auction}. A historical motivation is the leasing of
the Laurion silver mines in classical Athens, where inscriptions record leases
and public sales administered by the \emph{poletai}
~\cite{crosby1950leases,lalonde1991agora,aperghis1998reassessment}. From a
modern perspective, a lease had context-dependent surplus: geology, prior
excavation, nearby discoveries, technology, costs, and silver demand all
mattered. A highest-bid lease sale could raise public revenue, but it need not
allocate the right to the operator who would generate the greatest total
surplus. The counterfactual design question is therefore: how should a public
operator allocate resources when value must be learned from experience?

Laurion is only a historical motivation; our model is a platform repeatedly
procuring a context-dependent service. Electricity markets are the leading application:
operators procure demand response, reserves, and reliability services under
changing weather, load, renewable output, and congestion. The spot-pricing and
network-pricing literature treats efficient dispatch as social-surplus or
system-cost optimization
~\cite{bohn1984optimal,schweppe1988spot,hogan1992contract}, while bandit demand
response adds learning~\cite{jain2020multiarmed}. Here \(\mu_{i,c}\) is the
expected system value delivered by resource \(i\) in context \(c\), and \(k_i\)
is its privately reported standing activation cost. The same structure covers
LLM-era compute allocation, where the value of GPU or inference capacity depends on task,
latency, quality, and provider costs~\cite{ghodsi2011dominant,patra2026truthful},
and environmental procurement of abatement, conservation, or carbon removal
~\cite{coase1960problem,dales1968pollution,montgomery1972markets,cramton2002tradeable,latacz1997auctioning}.

We study a repeated contextual procurement auction that captures the procurement
side of these allocation problems. There are \(n\) producers and a finite set of
contexts. Producer \(i\) has a private cost \(k_i\). When producer \(i\) is
selected in context \(c\), the platform obtains an unknown gross value with mean
\(\mu_{i,c}\), observes only the selected producer-context outcome, and learns
from this bandit feedback. Under truthful bidding, the social surplus of
selecting \(i\) in context \(c\) is \(\mu_{i,c}-k_i\), and the outside option has
surplus zero. The full-information benchmark is the efficient procurement
mechanism that knows \(\mu\) and, in each context, selects the producer with the
largest nonnegative surplus. Our goal is to design learning mechanisms with low
welfare regret while controlling strategic incentives.

\smallskip
\noindent{\bf Our contributions.}
We formulate contextual procurement with private costs and unknown values,
using total-surplus regret so the bounds have a direct welfare interpretation.
We show that UCB allocation achieves \(\widetilde O(\sqrt{ngT})\) regret under
truthful bidding. We then give a bid-independent explore-then-commit mechanism
with empirical critical payments; it is dominant-strategy truthful and has
\(\widetilde O((ng)^{1/3}T^{2/3})\) regret. Finally, we introduce frozen-payment
UCB, which separates bid-independent payment learning from adaptive allocation
learning. Under a truthful-path margin condition, its near-UCB tuning obtains
\(\widetilde O(\sqrt{ngT})\) regret and average per-round deviation gain
\(\widetilde O(T^{-1/4})\) for fixed \(n,g\); the balanced tuning is in the
supplement. For the same post-exploration benchmark, our lower bounds match
the horizon exponents over the frozen critical-payment class.

\section{Related Work}

\smallskip
\noindent{\bf Auction design and welfare benchmarks.}
Our full-information benchmark is the efficient single-parameter procurement
mechanism, implemented by critical payments in the VCG tradition
~\cite{vickrey1961counterspeculation,clarke1971multipart,groves1973incentives};
see~\cite{nisan2007algorithmic}. This welfare benchmark differs from Myerson's
revenue benchmark, where virtual values or virtual costs characterize optimal
mechanisms~\cite{myerson1981optimal}. Our regret is measured directly in
surplus \(\mu_{i,c}-k_i\), not in transformed virtual costs; the supplement
records related procurement and approximate-IC work.

\smallskip
\noindent{\bf Bandit learning and strategic allocation.}
The learning side is closest to stochastic bandits
~\cite{robbins1952some,lai1985asymptotically,auer2002finite,bubeck2012regret}
and contextual bandits~\cite{langford2008epoch,lu2010contextual}.
In strategic bandit allocation, the learning path depends on bids, making
monotonicity and payment computation delicate. Prior work characterizes truthful
multi-armed-bandit mechanisms and their costs, and uses resampling to obtain
truthful-in-expectation mechanisms
~\cite{babaioff2014characterizing,devanur2009price,babaioff2015truthful}. Our
route is different: ETC uses bid-independent exploration for exact truthfulness,
while F-UCB freezes payment estimates to obtain approximate truthfulness.
Related incentive-aware bandit mechanisms appear in crowdsourcing,
expertsourcing, and demand response
~\cite{jain2016deterministic,jain2018quality,jain2020multiarmed}. Strategic
combinatorial bandit mechanisms are also related
~\cite{chen2013combinatorial,gao2021auction}, but these works do not isolate
the frozen critical-payment tradeoff studied here.

\smallskip
\noindent{\bf Contextual auction learning.}
Recent contextual auction-learning papers study sponsored-search mechanisms,
revenue-regret objectives, or bidder learning in fixed auction formats
~\cite{abhishek2020designing,zhang2024online,han2025optimal}. Our setting
instead asks the platform to design allocation and payment rules while
controlling welfare regret and producer incentives. The closest recent work is
Patra, Damle, Padala, and Gujar, who study truthful reverse auctions for
adaptive LLM provider selection via contextual MABs~\cite{patra2026truthful}.\footnote{
Two proof steps in~\cite{patra2026truthful} appear to require additional
justification. Claim 1 uses bid-independence of value estimates, but forced
exploration depends on an active set computed from bid-dependent OVS scores.
Claim 2 appears to need a weaker active-set monotonicity statement, since
lowering one provider's virtual cost can increase the maximum OVS and exclude
other providers. These gaps do not rule out an approximate-truthfulness reading,
but they leave exact truthfulness of the stated algorithm unestablished by the
current proof.}
Their regret is defined relative to a Myerson-style virtual-cost benchmark. This
is natural for that mechanism-design objective, but virtual-cost regret is not a
realized welfare, payment, or buyer-surplus loss because virtual costs include
distributional information-rent corrections. We instead target welfare regret
relative to the efficient surplus benchmark, as is appropriate when the operator
acts as a planner and payments are instruments for elicitation and coordination.
A virtual-cost analogue under additional regularity is recorded in the supplement.

\section{Model and Full-Information Benchmark}

We study a repeated contextual procurement problem with unknown product
values. There are $n$ producers, indexed by $i\in[n]$. Producer $i$ owns a
single product $a_i$, whose context-dependent gross value is unknown to the
platform and must be learned from bandit feedback.

Producer $i$ has private production cost $k_i\in[0,1]$. At the beginning of
the interaction, producer $i$ submits an ask bid $b_i\in[0,1]$, which remains
fixed throughout the horizon. Truthful bidding means $b_i=k_i$.

There is a finite context set $\C=\{1,\dots,g\}$.
At each round $t=1,\dots,T$, a context $c_t\in\C$ is observed. The context
process is exogenous and action-independent. It may be stochastic and depend
on prior contexts, but its conditional law is unaffected by bids, allocations,
reward realizations, or mechanism randomization. In particular, changing a
producer's report does not change the law of future contexts.

Let $\mathcal F_{t-1}$ be the history immediately before context $c_t$ is
observed. If product $a_i$ is selected in context $c$, the platform observes
a bounded random outcome $Y_t(i)\in[0,1]$ with
$\E[Y_t(i)\mid \mathcal F_{t-1},c_t=c]=\mu_{i,c}$.
Thus, for every adaptively selected producer--context pair, the centered
observations form a bounded martingale-difference sequence. This conditional
mean assumption, rather than a mean conditional only on the current context,
is what permits the uniform UCB concentration bounds below. The unknown mean
matrix is $\mu=\{\mu_{i,c}\}_{i\in[n],c\in\C}$. The platform observes only
bandit feedback: after choosing $i_t$, it
observes $Y_t(i_t)$ and does not observe the counterfactual outcomes of other
products.

If producer $i$ is selected and paid $p_{i,t}$, its realized utility is
$p_{i,t}-k_i$. The realized social surplus contribution is $Y_t(i)-k_i$.

The learning objective in this paper is welfare-maximizing efficient
procurement under truthful costs. Payments are transfers between the platform
and producers and therefore cancel from social surplus; they are used to control
incentives, while regret is measured against the full-information efficient
allocation.

Throughout the paper, all tie-breaking rules are fixed ex ante and independent
of the submitted bids.

For a mechanism $\mathcal M$, let
$\mathcal U_i^{\mathcal M}(b_i,b_{-i};k_i)$ denote producer $i$'s total
cumulative utility\sadie{cumulative?}\shi{fixed} over the $T$ rounds when its true cost is $k_i$ and it
reports $b_i$. Expectations are over the context process, reward noise, and
mechanism randomness. \sadie{I would slightly prefer to define approximate truthfulness earlier in the model section and we can add ``exact truthfulness is when $\epsilon=0$''.}\shi{fixed}
\begin{definition}[Approximate truthfulness]
\label{def:approx-truth}
Fix producer $i$, true cost $k_i$, and other bids $b_{-i}$. A mechanism
$\mathcal M$ is $\epsilon$-approximately truthful at
$(i,k_i,b_{-i})$ if, for every fixed report $b_i'\in[0,1]$,
\[
    \E\!\left[
        \mathcal U_i^{\mathcal M}(k_i,b_{-i};k_i)
    \right]
    \ge
    \E\!\left[
        \mathcal U_i^{\mathcal M}(b_i',b_{-i};k_i)
    \right]
    -
    \epsilon .
\]
\end{definition}
Exact truthfulness is the special case $\epsilon=0$, namely dominant-strategy
truthfulness (DSIC). For $\epsilon>0$, this is a pointwise
$\epsilon$-approximately truthful, or approximate-DSIC, guarantee rather than
one averaged over types; setting $b_{-i}=k_{-i}$ further yields an
$\epsilon$-approximate ex-post Nash equilibrium. This is the usual additive
relaxation of exact incentive constraints used in the approximate-IC literature;
see, for example, \cite{balseiro2024mechanism}. The expectation is only over
the context process, reward noise, and mechanism randomness. Thus $\epsilon$ is
a cumulative deviation gain, and $\epsilon/T$ is the corresponding average
per-round incentive.

\subsection{Full-Information Benchmark}

We first define the benchmark mechanism that knows the true value matrix
$\mu$. For a bid profile $b=(b_1,\dots,b_n)$ and context $c$, define producer
$i$'s score as $S_i(c;b)=\mu_{i,c}-b_i$. The outside option $\empt$ has score
$S_{\empt}(c;b)=0$.
The full-information efficient procurement allocation rule chooses
\[
    a^\star(c;b)
    \in
    \arg\max_{a\in[n]\cup\{\empt\}} S_a(c;b).
\]
Equivalently, the platform procures from a producer only when the largest
score is nonnegative.

For producer $i$, define the competing threshold and the
full-information VCG critical payment by
\[
    H_i(c;b_{-i})
    =
    \max\left\{
        0,
        \max_{j\neq i}
        \left(\mu_{j,c}-b_j\right)
    \right\},
\]
and
\[
    q_i^\star(c;b_{-i})
    =
    \left[
        \mu_{i,c}-H_i(c;b_{-i})
    \right]_{[0,1]},
\]
where $[x]_{[0,1]}=\min\{1,\max\{0,x\}\}$\sadie{ and $q_i^\star(c;b_{-i})$ equal to the highest payment at which selecting $i$ remains competitive with the platform’s best alternative}.\shi{fixed}
Thus, $q_i^\star(c;b_{-i})$ is the highest payment, equivalently reported
cost, at which selecting $i$ remains competitive with the platform's best
alternative.
If producer $i$ is selected, the platform pays $q_i^\star(c;b_{-i})$. If
producer $i$ is not selected, it receives zero.
\sadie{Maybe we can add an half-sentence intuition for competing threshold and VCG critical payment since $[x]_{[0,1]}=\min\{1,\max\{0,x\}\}$ can be confusing and not everyone is familiar with VCG? } \shi{fixed}

Because the score $\mu_{i,c}-b_i$ is strictly decreasing in $b_i$ and
tie-breaking is fixed ex ante, the allocation rule is monotone nonincreasing in
$b_i$: a producer who reports a higher cost is weakly less likely to be
selected. The critical-payment rule therefore implements the full-information
efficient allocation truthfully, as in the VCG/critical-value tradition of
Vickrey~\cite{vickrey1961counterspeculation},
Clarke~\cite{clarke1971multipart}, and Groves~\cite{groves1973incentives}.

\begin{proposition}[Full-information truthfulness and individual rationality]
\label{prop:oracle-dsic}
The full-information efficient allocation rule with the critical payments above
is dominant-strategy truthful and ex-post individually rational.
\end{proposition}

Under truthful bidding, the full-information benchmark is
\[
    \OPT(T)
    =
    \sum_{t=1}^T
    \max\left\{
        0,
        \max_{i\in[n]}
        \left(\mu_{i,c_t}-k_i\right)
    \right\}.
\]
For any learning allocation rule that selects
$i_t\in[n]\cup\{\empt\}$ at round $t$, define its regret under
truthful bidding by the convention that the selected surplus is zero when
\(i_t=\empt\):
\[
\begin{aligned}
    B_t
    &:=
    \max\left\{0,\max_{i\in[n]}(\mu_{i,c_t}-k_i)\right\},\\
    \Reg(T)
    &:=
    \sum_{t=1}^T
    \left[
        B_t
        -
        \1\{i_t\neq\empt\}
        (\mu_{i_t,c_t}-k_{i_t})
    \right].
\end{aligned}
\]

This is contextual pseudo-regret because it is written in terms of the mean
values $\mu_{i,c}$. Under the conditional-mean assumption, its expectation equals
the expected realized welfare loss relative to the full-information efficient
allocation: reward noise has conditional mean zero, and payments are transfers.

\section{Explore-Then-Commit}
\label{sec:etc}

We first define an explore-then-commit (ETC) mechanism. It separates learning
from strategic allocation: the platform explores using a bid-independent rule,
freezes the resulting value estimates, and then runs the efficient critical-price
mechanism for the empirical instance. This gives a slower regret rate than UCB
allocation, but it gives exact dominant-strategy truthfulness. The complete
mechanism is given in Algorithm~\ref{alg:etc}.

Let $M\in\{1,\dots,T\}$ be the exploration length.

\begin{algorithm}[t]
\caption{Explore-then-commit procurement}
\label{alg:etc}
\begin{algorithmic}[1]
\STATE Set \(M=\min\{T,\lceil(ng)^{1/3}T^{2/3}\rceil\}\) unless specified.
\FOR{\(t=1,\ldots,M\)}
    \STATE Explore bid-independently; pay \(1\) and observe \(Y_t(i_t)\).
\ENDFOR
\STATE Compute \(\wh\mu^0\) and the empirical critical prices \(\wh q_i^0(c;b_{-i})\).
\FOR{\(t=M+1,\ldots,T\)}
    \STATE Choose the empirical-score maximizer, with outside score \(0\).
    \STATE If producer \(i\) is selected, pay \(\wh q_i^0(c_t;b_{-i})\); otherwise pay \(0\).
\ENDFOR
\end{algorithmic}
\end{algorithm}

\emph{Exploration phase.} ETC uses a bid-independent rule for the first $M$
rounds and pays each selected producer $1$ (or any bid-independent constant
$\ge\max_i k_i$). Let $N^0_{i,c}$ be the number of observations of $(i,c)$ and
$\wh\mu^0_{i,c}$ its empirical mean. These payments ensure exploration-phase
individual rationality without affecting incentives.

\emph{Commit phase.} ETC freezes $\wh\mu^0$. In each round $t>M$, define
$\wh S_i^0(c;b)=\wh\mu^0_{i,c}-b_i$ for $i\in[n]$ and
$\wh S_{\empt}^0(c;b)=0$. The mechanism chooses
$i_t\in\arg\max_{a\in[n]\cup\{\empt\}}\wh S_a^0(c_t;b)$. For the selected
producer $i_t=i$, define
$\wh H_i^0(c_t;b_{-i})=\max\{0,\max_{j\neq i}(\wh\mu^0_{j,c_t}-b_j)\}$ and
pay the empirical VCG critical payment
$\wh q^0_i(c_t;b_{-i})=[\wh\mu^0_{i,c_t}-\wh H_i^0(c_t;b_{-i})]_{[0,1]}$.
If no producer is selected, no payment is made.

\begin{assumption}[Uniform exploration coverage]
\label{ass:procurement-coverage}
There exists a constant $c_{\mathrm{cov}}>0$ such that, with probability at
least $1-(Tng)^{-2}$,
\[
    N^0_{i,c}
    \ge
    c_{\mathrm{cov}}\frac{M}{ng}
    \qquad
    \text{for all }(i,c)\in[n]\times\C.
\]
This is a joint coverage condition on the exogenous context process and the
bid-independent exploration policy: during the exploration phase, each context
must occur sufficiently often, and conditional on each context, the exploration
policy must sample each producer sufficiently often. In particular, this
assumption is nonvacuous only when $M$ is at least of order $ng$ (and typically
$ng\log(Tng)$ for high-probability coverage).
\end{assumption}

A standard i.i.d.-context and uniform-exploration condition that implies
Assumption~\ref{ass:procurement-coverage} by Chernoff bounds is recorded in the
supplement.

\begin{theorem}[Explore-then-commit regret]
\label{thm:etc-procurement-regret}
Under the conditional reward assumption in the model,
Assumption~\ref{ass:procurement-coverage}, and truthful bidding, the
explore-then-commit procurement mechanism satisfies
\[
    \E\!\left[\Reg(T)\right]
    =
    \wt O\!\left(
        M
        +
        T\sqrt{\frac{ng}{M}}
    \right).
\]
With \(M\asymp(ng)^{1/3}T^{2/3}\), the bound becomes
\(\E[\Reg(T)]=\wt O((ng)^{1/3}T^{2/3})\). Thus, for fixed
\(n\) and \(g\), ETC has \(\wt O(T^{2/3})\) welfare regret.
\end{theorem}

\begin{theorem}[Truthfulness of explore-then-commit]
\label{thm:etc-procurement-truthfulness}
Suppose the exploration allocation rule and exploration payments are
bid-independent, the commit-phase allocation and payments are computed from the
frozen exploration estimates $\wh\mu^0$, and tie-breaking is fixed ex ante and
bid-independent. Then the explore-then-commit procurement mechanism is
dominant-strategy truthful. With the exploration payment set to $1$, it is also
ex-post individually rational.
\end{theorem}

\section{Frozen-Payment UCB}
\label{sec:ucb-frozen-payments}

We retain UCB's adaptive allocation while computing payments only from
bid-independent data. Frozen-payment UCB freezes payment estimates before later
bid-dependent allocations can affect pricing evidence. We first recall the UCB
allocation rule used after exploration.

For each pair $(i,c)$, let $N_{i,c}(t)$ be the number of observations of
producer $i$ in context $c$ before round $t$, and let $\wh\mu_{i,c}(t)$ be the
corresponding empirical mean. We use the convention $\wh\mu_{i,c}(t)=0$ whenever
$N_{i,c}(t)=0$. Define
\[
    \beta_{i,c}(t)
    =
    \sqrt{
        \frac{C\log(2Tng)}
             {N_{i,c}(t)\vee1}
    },
\]
where $C\ge2$ is a sufficiently large universal constant. Unobserved pairs are
optimistic by construction, and on observed pairs optimism holds on the UCB
concentration event below. We state the results for known horizon $T$; a
standard doubling schedule removes this knowledge at an additional logarithmic
factor.

Given bids $b$, define the optimistic UCB scores by
$\wt S_i(c;b,t)=\wh\mu_{i,c}(t)+\beta_{i,c}(t)-b_i$ for $i\in[n]$ and
$\wt S_{\empt}(c;b,t)=0$. The UCB procurement allocation rule chooses
$i_t^{\UCB}(b)\in\arg\max_{a\in[n]\cup\{\empt\}}\wt S_a(c_t;b,t)$.
\sadie{Maybe define this as $i_t^{\text{UCB}}(b)$ since we use it in Lemma 1.}\shi{fixed}

\begin{proposition}[UCB regret]
\label{thm:ucb-procurement-regret}
Suppose the bounded conditional-mean assumption in the model holds for every
adaptive history and producer--context pair. Under truthful bidding \(b_i=k_i\),
the UCB procurement allocation rule satisfies
\[
    \E\!\left[\Reg(T)\right]
    \le
    4\sqrt{CngT\log(2Tng)}
    +O\!\left(\frac{1}{Tn^2g^2}\right).
\]
In particular, if $n$ and $g$ are constants, then
$\E[\Reg(T)]=\wt O(T^{1/2})$.
\end{proposition}

UCB gives the sharper welfare-regret benchmark, but a bid can affect early
allocations, later samples, estimates, bonuses, and allocations. Thus dynamic
truthfulness does not follow from per-round critical payments.

We analyze UCB with frozen payments: payment data and estimates are
bid-independent, though critical payments may depend on other producers' bids.
After $M$ bid-independent exploration rounds (paying $1$, say), the mechanism
freezes $\wh\mu^{\mathrm{pay}}$ for payments; later observations update allocation
estimates only. Allocation estimates may reuse exploration samples, but payment
estimates are never updated.

When producer $i$ is selected in context $c$, it receives the frozen critical
payment $\wh q_i(c;b_{-i})$, which depends on frozen estimates, context, and
other bids, but not on $i$'s own bid. Its bid can still affect selection and the
future UCB history, which is why F-UCB is only approximately truthful.

\subsection{Frozen Payment Rule and Regret}

Given frozen payment estimates $\wh\mu^{\mathrm{pay}}$, define
\[
\begin{aligned}
\wh H_i(c;b_{-i})
&=
\max\{0,\max_{j\neq i}(\wh\mu^{\mathrm{pay}}_{j,c}-b_j)\},\\
\wh q_i(c;b_{-i})
&=
[\wh\mu^{\mathrm{pay}}_{i,c}-\wh H_i(c;b_{-i})]_{[0,1]} .
\end{aligned}
\]
If producer $i$ is selected in context $c$, it is paid
$\wh q_i(c;b_{-i})$.

At round $t>M$, allocation uses the adaptive scores
\[
\wt S_i(c_t;b,t)=\wh\mu_{i,c_t}(t)+\beta_{i,c_t}(t)-b_i,
\qquad
\wt S_{\empt}(c_t;b,t)=0 .
\]
It then chooses
$i_t=i_t^{\UCB}(b)\in\arg\max_{a\in[n]\cup\{\empt\}}\wt S_a(c_t;b,t)$.
Later observations update allocation estimates only, never frozen payments;
the complete mechanism is Algorithm~\ref{alg:fucb}.

\begin{algorithm}[t]
\caption{Frozen-payment UCB procurement}
\label{alg:fucb}
\begin{algorithmic}[1]
\STATE Choose a bid-independent payment-exploration length \(M\).
\FOR{\(t=1,\ldots,M\)}
    \STATE Explore bid-independently; pay \(1\) and observe \(Y_t(i_t)\).
\ENDFOR
\STATE Freeze payment estimates and critical prices; initialize UCB estimates.
\FOR{\(t=M+1,\ldots,T\)}
    \STATE Choose the UCB-score maximizer, with outside score \(0\).
    \STATE If selected, pay \(\wh q_i(c_t;b_{-i})\), observe reward, and update UCB estimates.
\ENDFOR
\end{algorithmic}
\end{algorithm}

The main text uses \emph{near-UCB tuning},
$M\asymp\sqrt{ngT}$ (capped at $T$), which preserves the UCB regret order;
the supplement gives the balanced tuning and derivations.

\begin{theorem}[Regret of frozen-payment UCB]
\label{thm:fucb-regret}
Under truthful bidding and the bounded conditional-mean assumption, the
frozen-payment UCB mechanism with any $M$-round bid-independent exploration
phase, whose allocation estimates are initialized with the exploration
observations, satisfies
\[
\E\!\left[\Reg(T)\right]
=
\wt O\!\left(M+\sqrt{ngT}\right).
\]
In particular, for fixed $n$ and $g$, the near-UCB tuning achieves
$\E[\Reg(T)]=\wt O(T^{1/2})$. The balanced specialization is in the
supplement.
\end{theorem}

\begin{assumption}[Payment exploration coverage]
\label{ass:payment-regularity}
The bid-independent exploration phase used to compute the frozen payments
satisfies the uniform coverage condition in
Assumption~\ref{ass:procurement-coverage}.
\end{assumption}

\begin{lemma}[Frozen critical payments are accurate]
\label{lem:frozen-critical-payment-accuracy}
Under Assumption~\ref{ass:payment-regularity}, with probability at least
$1-2(Tng)^{-2}$,
$\varepsilon_M:=\max_{i,c}|\wh\mu^{\mathrm{pay}}_{i,c}-\mu_{i,c}|
=\wt O(\sqrt{ng/M})$.
On this event, for every producer $i$, context $c$, and other bids $b_{-i}$,
\[
    \left|
        \wh q_i(c;b_{-i})-q_i^\star(c;b_{-i})
    \right|
    \le
    2\varepsilon_M
    =
    \wt O\!\left(
        \sqrt{\frac{ng}{M}}
    \right).
\]
\end{lemma}

\subsection{Truthful-Path Allocation Stability}

We use the following stability condition.

\begin{assumption}[Truthful-path margin]
\label{ass:truthful-path-margin}
Fix a producer $i$ and other bids $b_{-i}$. Let $b^0=(k_i,b_{-i})$ be the bid
profile in which producer $i$ bids truthfully. There exists a
possibly horizon-dependent margin $\rho_T>0$ such that for every context
$c\in\C$, the full-information optimal action
$a^\star(c;b^0)\in\arg\max_{a\in[n]\cup\{\empt\}}S_a(c;b^0)$ is unique and
satisfies
\[
    S_{a^\star(c;b^0)}(c;b^0)
    -
    \max_{a\neq a^\star(c;b^0)}
    S_a(c;b^0)
    \ge
    \rho_T .
\]
\end{assumption}

\begin{remark}[When the truthful-path margin holds]
\label{rem:when-margin-holds}
In smoothed or continuous models, small gaps are rare. For example, suppose
that, for every fixed context and every pair of distinct actions
$a,a'\in[n]\cup\{\empt\}$, the score difference
$S_a(c;b^0)-S_{a'}(c;b^0)$ induced by random costs and values has density at
most $B$ in a neighborhood of zero. Then
\[
    \Prb\left\{
        |S_a(c;b^0)-S_{a'}(c;b^0)|\le\rho
    \right\}
    \le 2B\rho .
\]
A union bound over contexts and action pairs gives probability
$O(Bg(n+1)^2\rho)$ that some pairwise gap is at most $\rho$. For the near-UCB
tuning, define
\[
    \rho_{T,\mathrm{near}}^\star
    =
    C_{\mathrm{gap}}\sqrt{ng\log(2Tng)}\,T^{-3/8}.
\]
This scale makes the truthful-path mismatch term no larger than the
payment-estimation term, with exceptional probability
$O(Bg(n+1)^2\rho_{T,\mathrm{near}}^\star)$. When $n$ or $g$ grows, the displayed
dependence should be kept explicitly. The analogous balanced scale is in the
supplement.
\end{remark}

\begin{lemma}[Truthful UCB allocation stability]
\label{lem:truthful-ucb-stability}
Suppose Assumption~\ref{ass:truthful-path-margin} holds. On the UCB
concentration event
$\mathcal E=\{\forall i,c,t:
|\wh\mu_{i,c}(t)-\mu_{i,c}|\le\beta_{i,c}(t)\}$, the truthful UCB path satisfies
\[
    \sum_{t=M+1}^T
    \1\left\{
        i_t^{\UCB}(b^0)
        \neq
        a^\star(c_t;b^0)
    \right\}
    \le
    \wt O\!\left(\frac{ng}{\rho_T^2}\right).
\]
More generally, for any $\alpha\in(0,1/2]$, if
$\rho_T\ge C_{\mathrm{gap}}\sqrt{ng\log(2Tng)}\,T^{-\alpha}$ for a sufficiently
large constant $C_{\mathrm{gap}}$, then the number of post-exploration
allocation mistakes is at most $O(T^{2\alpha})$.
\end{lemma}

\subsection{Approximate Truthfulness}

Let $H_M$ be the bid-independent exploration history used to freeze payments.
For the comparison below, we also consider the producer's payoff after
observing $H_M$ and choosing a single post-exploration report. The action
$\bot$ denotes withdrawal from the post-exploration phase and yields zero
post-exploration utility. Equivalently, its post-exploration selection
indicator is fixed to zero.

For fixed $i$, $k_i$, and $b_{-i}$, write
\[
\bar U_i(b;H_M)
:=
\E\!\left[
\mathcal U_i^{\FUCB}(b,b_{-i};k_i)
\mid H_M
\right],
\]
Here, $b$ is used only after $H_M$; the bid-independent exploration history and
its utility are held fixed, so this comparison also covers fixed upfront reports.
We then define
\[
\Delta_i(b_i';H_M)
:=
\bar U_i(b_i';H_M)-\bar U_i(k_i;H_M).
\]
\[
\epsilon_{i,T}^{\mathrm{post}}(k_i;b_{-i})
:=
\E_{H_M}\!\left[
\sup_{b_i'\in[0,1]\cup\{\bot\}}
\left(\Delta_i(b_i';H_M)\right)_+
\right].
\]
The bid-independent exploration utility cancels in this comparison. This
auxiliary benchmark treats $H_M$ as observed by the producer.

For a realized $H_M$, write
\[
\eta(H_M):=\max_c|\wh q_i(c;b_{-i})-q_i^\star(c;b_{-i})|
\]
and let
\[
Z(H_M):=\E\!\left[
\sum_{t=M+1}^T
\1\!\left\{
\begin{gathered}
i_t^{\UCB}(k_i,b_{-i})\\
\neq a^\star(c_t;(k_i,b_{-i}))
\end{gathered}
\right\}
\middle| H_M
\right].
\]

\sadie{I think it would be helpful to add some sentence to highlight the intuition behind this bound: we are unifying the payment error (Lemma 1, $\eta(H_M)$) and truthful allocation mistake (Lemma 2, $Z(H_M)$). }\shi{fixed}
The proof compares both post-exploration paths with the same full-information
truthful benchmark. Conditional on $H_M$, a deviation has utility at most the
benchmark plus $(T-M)\eta(H_M)$, whereas truthful UCB has utility at least the
benchmark less by $(T-M)\eta(H_M)+Z(H_M)$. Thus payment error affects both
paths, while allocation mistakes affect only the truthful path; subtracting
the two comparisons gives the next lemma.
\begin{lemma}[Report-uniform comparison]
\label{lem:report-uniform-comparison}
For every realized exploration history,
\[
\sup_{b_i'\in[0,1]\cup\{\bot\}}
\left(\Delta_i(b_i';H_M)\right)_+
\le 2(T-M)\eta(H_M)+Z(H_M).
\]
\end{lemma}

\begin{theorem}[Approximate truthfulness from truthful-path stability]
\label{thm:ucb-direct-approx-truth}
Fix producer $i$ and other bids $b_{-i}$. Let $b^0=(k_i,b_{-i})$. Suppose the
frozen-payment UCB mechanism uses an exploration phase of length $M$ whose
allocations and payments are bid-independent, and the frozen payment estimates
are computed only from this exploration data. Under
Assumptions~\ref{ass:payment-regularity} and
\ref{ass:truthful-path-margin}, for every deviation $b_i'\in[0,1]$,
\[
\begin{aligned}
    &\E\!\left[
        \mathcal U_i^{\FUCB}(b_i',b_{-i};k_i)
        -
        \mathcal U_i^{\FUCB}(k_i,b_{-i};k_i)
    \right]\\
    &\qquad\le
    \wt O\!\left(
        \frac{ng}{\rho_T^2}
        +
        T\sqrt{\frac{ng}{M}}
    \right).
\end{aligned}
\]
Equivalently, F-UCB is \(\epsilon_T\)-approximately truthful at
\((i,k_i,b_{-i})\) in the sense of Definition~\ref{def:approx-truth}, where
\[
\epsilon_T
=
\wt O\!\left(
\frac{ng}{\rho_T^2}
+
T\sqrt{\frac{ng}{M}}
\right).
\]
Consequently, the average per-round incentive to deviate is
\[
\begin{aligned}
    &\frac{1}{T}
    \E\!\left[
        \mathcal U_i^{\FUCB}(b_i',b_{-i};k_i)
        -
        \mathcal U_i^{\FUCB}(k_i,b_{-i};k_i)
    \right]\\
    &\qquad\le
    \wt O\!\left(
        \frac{ng}{T\rho_T^2}
        +
        \sqrt{\frac{ng}{M}}
    \right).
\end{aligned}
\]
For fixed $n$ and $g$, the near-UCB tuning gives
\(\epsilon_T=\wt O(T^{3/4})\) and average incentive \(\wt O(T^{-1/4})\) under
\(\rho_T\ge\wt\Omega(T^{-3/8})\). The balanced incentive specialization is in
the supplement.
\end{theorem}

\begin{corollary}[Smoothed-instance approximate truthfulness]
\label{cor:smoothed-approx-truth}
Under the smoothed density condition in
Remark~\ref{rem:when-margin-holds} with bound \(B\), and
Assumption~\ref{ass:payment-regularity}, F-UCB with the near-UCB tuning is
\(\wt O(T^{3/4})\)-approximately truthful, with average per-round incentive
\(\wt O(T^{-1/4})\), at \((i,k_i,b_{-i})\) for fixed \(n,g\), with instance
probability at least
\[
    1-O\!\left(Bg(n+1)^2\rho_{T,\mathrm{near}}^\star\right).
\]
Conditional on the realized instance, utility expectations are over contexts,
reward noise, and mechanism randomness. The balanced smoothed specialization is
in the supplement.
\end{corollary}

\begin{corollary}[Unilateral deviation under truthful opponents]
\label{cor:fucb-truthful-opponents}
The same approximate-truthfulness bound in
Theorem~\ref{thm:ucb-direct-approx-truth} holds for unilateral deviations
against truthful opponents, taking \(b_{-i}=k_{-i}\).
\end{corollary}

\begin{theorem}[Post-exploration benchmark]
\label{thm:ucb-post-exploration}
Under the assumptions of
Theorem~\ref{thm:ucb-direct-approx-truth},
\[
\epsilon_{i,T}^{\mathrm{post}}(k_i;b_{-i})
\le
\wt O\!\left(
\frac{ng}{\rho_T^2}
+
T\sqrt{\frac{ng}{M}}
\right).
\]
The bound permits the producer to choose its report after observing $H_M$.
For fixed \(n,g\), the near-UCB tuning gives
\(\epsilon_{i,T}^{\mathrm{post}}=\wt O(T^{3/4})\) under
\(\rho_T\ge\wt\Omega(T^{-3/8})\); the balanced post-exploration specialization
is in the supplement.
\end{theorem}

The corresponding approximate individual-rationality interpretation is in the
supplement.

\section{Lower Bound and Matching Guarantees}
\label{sec:lower-bound}

We establish the lower bound in a single-context, single-producer Bernoulli
instance. Selecting the producer yields an outcome with unknown mean \(\mu\).
Its full-information score is \(\mu-b\), so the benchmark selects it for
\(b<\mu\) and rejects it for \(b>\mu\), with fixed tie-breaking at equality;
the critical ask is \(q^\star(\mu)=\mu\).

Consider a frozen-payment mechanism that computes \(\widehat q(H_M)\) from a
bid-independent history \(H_M\) of at most \(M\) observations, possibly
collected at arbitrary rounds. For the lower bound, we grant the mechanism the
favorable timing in which all payment-learning data are available before the
remaining \(L:=T-M\) rounds: allocation may continue to adapt, but the payment
cannot be updated. For cost \(k\), let \(\epsilon_T^{\mathrm{post}}(k)\) be the
expected positive gain from the best fixed report or withdrawal after observing
\(H_M\), and let \(R_T^{\mathrm{post}}(k)\) be expected welfare regret over the
remaining \(L\) rounds under truthful bidding. Let \(R_T^{\mathrm{tot}}\) be the
worst-case expected welfare regret over the full horizon, with the worst case
taken over \(\mu\) and \(k\). The lower bound then separates two losses: the
\(M\) bid-independent payment-learning decisions
cost \(\Omega(M)\) regret, while an unavoidable payment-estimation error of
order \(\Theta(M^{-1/2})\) accumulates over the remaining \(L\) rounds. Across
two near-threshold costs,
this latter loss must appear in the aggregate of post-learning regret and
incentive error, yielding the frontier \(M+(T-M)/\sqrt M\).

\sadie{We now show that the tradeoff achieved by frozen-payment mechanisms is unavoidable. It suffices to establish the lower bound in a single-context, single-producer instance. Selecting the producer generates a Bernoulli outcome with unknown mean \(\mu\). If the producer reports cost \(b\), its score is \(\mu-b\), while the outside option has score \(0\). Thus, the full-information allocation selects the producer if and only if \(b\leq\mu\), and the corresponding critical payment is \(q^\star(\mu)=\mu\). The mechanism uses \(M\) bid-independent payment-learning decisions, producing a history \(H_M\) with at most \(M\) reward observations, from which it computes a frozen payment \(\widehat q(H_M)\). For the lower bound, we grant the mechanism the favorable timing in which \(H_M\) is observed before the remaining \(L:=T-M\) rounds; during these rounds, it may continue adapting its allocation rule, but not the payment.}\shi{fixed}

\sadie{The lower bound separates two distinct losses. First, the \(M\) bid-independent payment-learning decisions incur \(\Omega(M)\) worst-case full-horizon welfare regret, denoted \(R_T^{\mathrm{tot}}\). Second, any payment estimated from at most \(M\) observations has unavoidable error of order \(M^{-1/2}\); once the payment is frozen, this error accumulates over the remaining \(L\) rounds and must appear as either truthful-path post-learning welfare regret, \(R_T^{\mathrm{post}}(k)\), or post-learning incentive error, \(\epsilon_T^{\mathrm{post}}(k)\), defined as the producer's expected utility gain from its best deviation after observing \(H_M\). Together, these effects yield the frontier \(M+(T-M)/\sqrt{M}\).}\shi{fixed}

\begin{lemma}[Frozen critical-price estimation lower bound]
\label{lem:frozen-critical-price-lb}
There exist universal constants \(c_0,c_1>0\) such that, for every frozen
payment estimator \(\widehat q=\widehat q(H_M)\) based on at most \(M\)
bid-independent Bernoulli samples, there exists \(\mu\in[1/3,2/3]\) such
that, writing \(q^\star=\mu\) and \(\Delta=c_0/\sqrt M\), we have
\(\mathbb E_\mu[A(H_M)]\ge c_1\Delta\), where
\[
    A(H_M)
    :=
    \min\left\{
        \left(
            |\widehat q(H_M)-q^\star|-\frac{\Delta}{4}
        \right)_+,
        \frac{\Delta}{4}
    \right\}.
\]
\end{lemma}

\begin{theorem}[Frozen-payment regret--incentive tradeoff]
\label{thm:frozen-payment-lb}
For any frozen-payment mechanism as above, there exists
\(\mu\in[1/3,2/3]\) and two near-threshold costs
\(k^-=q^\star-\Delta/4\) and \(k^+=q^\star+\Delta/4\), where
\(q^\star=\mu\) and \(\Delta=\Theta(M^{-1/2})\),
such that
\[
\begin{aligned}
    &\epsilon_T^{\mathrm{post}}(k^-)
    +
    \epsilon_T^{\mathrm{post}}(k^+)
    +
    R_T^{\mathrm{post}}(k^-)
    +
    R_T^{\mathrm{post}}(k^+)\\
    &\qquad\ge
    \Omega\!\left(
        \frac{T-M}{\sqrt M}
    \right).
\end{aligned}
\]
Consequently, there exists a universal constant \(C>0\) such that, if
\[
    \epsilon_T^{\mathrm{post}}
    :=
    \sup_{\mu\in[0,1]}\sup_{k\in[0,1]}
    \epsilon_{T,\mu}^{\mathrm{post}}(k),
\]
where the dependence on \(\mu\) is made explicit only in this display, then
\[
    \epsilon_T^{\mathrm{post}}
    +
    C R_T^{\mathrm{tot}}
    \ge
    \Omega\!\left(
        \frac{T-M}{\sqrt M}
    \right).
\]
Thus, in $T$-only terms, this lower bound is $\Omega(T^{3/4})$ when
$M\asymp T^{1/2}$ and $\Omega(T^{2/3})$ when $M\asymp T^{2/3}$.
\end{theorem}

\begin{lemma}[Cost of bid-independent payment learning]
\label{lem:exploration-lb}
If a mechanism uses \(M\) payment-learning allocation decisions whose
distribution is independent of the submitted bid, then the worst-case expected
regret over the full horizon satisfies
\(R_T^{\mathrm{tot}}\ge\Omega(M)\).
Consequently, $M\asymp T^{1/2}$ already forces $\Omega(T^{1/2})$ regret, and
$M\asymp T^{2/3}$ forces $\Omega(T^{2/3})$ regret.
\end{lemma}

\begin{corollary}[Frozen-payment regret--incentive tradeoff]
\label{cor:frozen-payment-tradeoff-lb}
For any frozen-payment mechanism using \(M\) bid-independent payment-learning
samples as above, the combined incentive-regret cost obeys the following
frontier: there exist universal constants \(c,C>0\) such that
\[
    \epsilon_T^{\mathrm{post}}
    +
    C R_T^{\mathrm{tot}}
    \ge
    c\left(
        M+\frac{T-M}{\sqrt M}
    \right).
\]
In particular, for \(M\le T/2\), its right-hand side is
$\Omega(M+T/\sqrt M)$, which is minimized at
$M\asymp T^{2/3}$ with value $\Omega(T^{2/3})$.
Thus, within the frozen-payment framework, improving regret below
the explore-then-commit scale necessarily comes with a larger incentive term. In
particular,
any frozen-payment mechanism with
$R_T^{\mathrm{tot}}=\wt O(T^{1/2})$ must have
$\epsilon_T^{\mathrm{post}}=\wt\Omega(T^{3/4})$ up to logarithmic factors.
\end{corollary}

\sadie{These are exactly the horizon exponents attained by our upper bounds for fixed \(n,g\) under the same post-exploration incentive benchmark. Choosing \(M\asymp T^{2/3}\) matches the \(\widetilde O(T^{2/3})\) regret of ETC, which has zero incentive error, and the \(\widetilde O(T^{2/3})\) regret and incentive guarantees of balanced F-UCB. Choosing \(M\asymp T^{1/2}\) matches the \(\widetilde O(T^{3/4})\) post-exploration incentive error incurred by F-UCB while retaining near-\(\sqrt{T}\) truthful-path regret.}\shi{fixed}
For fixed \(n,g\), the frontier matches our upper bounds in horizon dependence
under the same post-exploration benchmark. ETC with \(M\asymp T^{2/3}\) has
\(\widetilde O(T^{2/3})\) regret and zero incentive error; under the stated
margin condition, the balanced F-UCB specialization in the supplement matches
this scale. With \(M\asymp T^{1/2}\), near-UCB retains
\(\widetilde O(\sqrt T)\) truthful-path regret and has
\(\widetilde O(T^{3/4})\) post-exploration incentive error.

The lower bound isolates the statistical source of the tradeoff. ETC commits
allocation as well as payments and is exactly truthful at the
explore-then-commit scale; F-UCB keeps learning after payments are frozen, and
frozen-payment error then creates the matching post-learning incentive cost.

\section{Experiments}
\label{sec:experiments}

We report two experiments on the fixed three-producer, three-context instance
given in the supplement. Contexts are sampled independently; exploration
feedback is Bernoulli with mean $\mu_{i,c}$, while post-exploration feedback is
set to its conditional mean to remove simulation noise. This remains a bounded
conditional-mean reward process. Error bands and parenthesized entries report
one standard error. Exploration selects a producer uniformly conditional on
context. We use UCB constant $C=2$ and fixed, bid-independent tie-breaking;
the supplement gives the remaining implementation details.

\begin{figure}[t]
    \centering
    \includegraphics[width=0.86\linewidth]{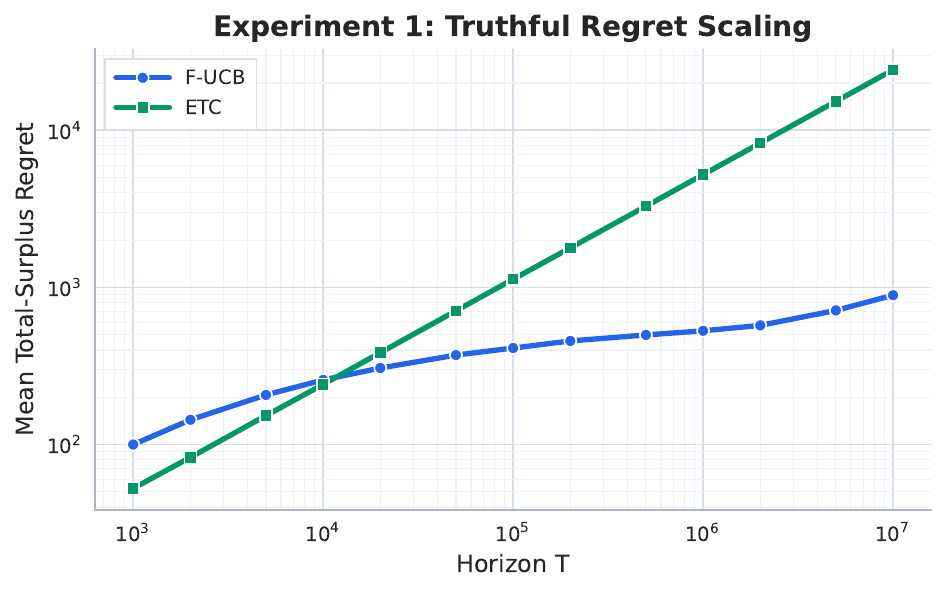}
    \caption{Truthful welfare regret on the fixed benchmark instance. ETC uses
    \(M=(ng)^{1/3}T^{2/3}\) bid-independent exploration; F-UCB uses
    \(M=T^{1/2}\) before UCB allocation with frozen payments.}
    \label{fig:exp-truthful-regret}
\end{figure}

Figure~\ref{fig:exp-truthful-regret} shows the expected qualitative pattern:
ETC pays a larger exploration cost, whereas F-UCB makes adaptive allocations
after shorter payment exploration and tracks the near-\(\sqrt T\) regret rate.

For strategic bidding, we run an ex-post F-UCB bid-grid search with
\(\epsilon=0.1\), which only sets grid spacing. This stress test gives producers
more information than the mechanism: for each simulation instance, the searched
profile may depend on realized paths. It is therefore a stress test rather than
an equilibrium computation; the supplement gives implementation details.
Table~\ref{tab:epsnet-equilibrium} compares its welfare regret with truthful
F-UCB and truthful ETC.

\begin{table}[t]
    \centering
    \begin{tabular}{@{}lccc@{}}
        \toprule
        & \multicolumn{2}{c}{F-UCB regret}
        & ETC regret\\
        \cmidrule(lr){2-3}\cmidrule(l){4-4}
        \(T\)
        & truthful
        & grid profile
        & truthful\\
        \midrule
        \(5\cdot 10^4\)
        & \(370.3\,(2.1)\)
        & \(804.6\,(324.4)\)
        & \(711.6\,(1.6)\)\\
        \(10^5\)
        & \(412.2\,(1.9)\)
        & \(650.6\,(226.0)\)
        & \(1123.0\,(2.0)\)\\
        \(2\cdot 10^5\)
        & \(451.1\,(2.3)\)
        & \(454.7\,(2.7)\)
        & \(1786.4\,(2.5)\)\\
        \(5\cdot 10^5\)
        & \(499.5\,(3.3)\)
        & \(502.0\,(3.4)\)
        & \(3288.6\,(2.4)\)\\
        \(10^6\)
        & \(528.3\,(3.3)\)
        & \(531.1\,(3.5)\)
        & \(5235.3\,(4.5)\)\\
        \bottomrule
    \end{tabular}
    \caption{Strategic bidding experiment using bid-grid spacing
    \(\epsilon=0.1\). Entries are means with one standard error over 40
    repetitions; ETC reports truthful regret, not bid-grid search.}
    \label{tab:epsnet-equilibrium}
\end{table}

Table~\ref{tab:epsnet-equilibrium} shows that strategic bids can matter at short
horizons, when frozen payments are noisy. As \(T\) grows, the bid-grid profile
approaches truthful F-UCB and remains far below ETC. This is descriptive
evidence of the finite-horizon payment/allocation feedback addressed by the
approximate-truthfulness analysis, not a separate equilibrium guarantee.

\section{Conclusion}

We study contextual procurement auctions in which a platform learns
context-dependent values while respecting producers' private costs. The results
make the learning--incentives tradeoff explicit. ETC is exactly truthful with
\(\widetilde O((ng)^{1/3}T^{2/3})\) welfare regret. Under the stated margin
condition, F-UCB with near-UCB tuning instead attains
\(\widetilde O(\sqrt{ngT})\) welfare regret and average per-round deviation gain
\(\widetilde O(T^{-1/4})\) for fixed \(n,g\).
ETC freezes both allocation and payment learning, whereas F-UCB freezes only
the payment evidence and continues to learn allocations adaptively.

Our lower bound matches the relevant horizon exponents under the same
post-exploration incentive benchmark within the frozen critical-payment class.
Thus, bid-independent payment learning supports strong welfare guarantees but
also exposes the cost of preserving incentives as allocation data adapt.
The supplement records the balanced tuning and discusses causal auction
learning and direct revenue regret. An important next step is to design payment
rules that safely use adaptive evidence without letting a producer's report
steer the evidence used to price it.
\bibliography{main}
\clearpage
\appendix
\section*{Appendix}

\section{Additional Discussion}
\label{app:additional-discussion}

\subsection{Future Directions}
One direction is to combine auction learning with peer prediction and other
tools for information elicitation without verification. Our model assumes that a
selected producer generates a conditionally unbiased reward observation. In many
procurement settings, however, the relevant quality signal may be delayed,
subjective, or only partially verifiable: reliability of a demand-response event,
usefulness of an LLM answer, or quality of an environmental project may be
assessed through peer reports, comparisons, or downstream audits rather than a
clean reward. Peer-prediction mechanisms for learning agents and comparison data
~\cite{feng2022peer,chen2024carrot} suggest ways to elicit such signals; the
auction-learning challenge is to integrate them with allocation and payment rules
while controlling both reporting incentives and welfare regret.

Another direction is causal auction learning. In many procurement settings, an
allocation is not just an observation but an intervention: activating a demand
response resource changes load, buying compute from one provider changes queueing
and latency, and environmental procurement can have spillovers across sites or
firms. Causal-bandit models~\cite{lattimore2016causal,lee2018structural},
including work on combinatorial interventions
~\cite{feng2023combinatorial,feng2025combinatorial}, the associated
pseudo-log-likelihood correction~\cite{feng2024correction}, and influence
propagation~\cite{feng2021causal}, offer tools for learning these intervention
effects while making allocation decisions.

A third direction is revenue maximization with economically interpretable
objectives. Our techniques can be adapted to virtual-cost allocation under
additional regularity, as discussed below, but virtual-cost regret is not itself
a realized welfare, payment, or buyer-surplus loss. A revenue-oriented extension
should therefore measure revenue regret directly, and then analyze how learning,
payment estimation, and incentive constraints interact under that objective.

\subsection{Virtual-Cost Variants}
The same template can be adapted to virtual costs by replacing
\(\mu_{i,c}-b_i\) with \(\mu_{i,c}-\psi_i(b_i)\), for a regular virtual-cost map
\(\psi_i\), and computing critical payments by inverting the virtual threshold.
To preserve quantitative approximate-truthfulness rates in bid space, one also
needs a quantitative regularity condition such as Lipschitzness of
\(\psi_i^{-1}\), or equivalently a strong monotonicity lower bound for
\(\psi_i\). Under such a condition, the same separation between payment
learning and allocation learning yields virtual-regret analogues of our
ETC/F-UCB bounds, with corresponding frozen-payment lower-bound exponents in
simple cases. We do not pursue this extension because virtual-cost regret is
not a primitive economic loss and is not the performance criterion optimized in
this paper.

\subsection{Additional Related Work}
Our single-unit model abstracts away from structured and capacitated procurement
constraints, where feasibility and multi-unit supply create additional
allocation structure~\cite{cary2008auctions,iyengar2008optimal}. Our
approximate-incentive guarantees are related to approximate mechanism design
~\cite{balseiro2024mechanism}, but the approximation here arises from online
learning and frozen payment estimation. On the learning side, applied and linear
contextual bandits provide standard benchmarks
~\cite{li2010contextual,chu2011contextual,abbasi2011improved}. Other
incentive-aware bandit mechanisms have been studied in crowdsourcing,
expertsourcing, and demand response
~\cite{jain2016deterministic,jain2018quality,jain2020multiarmed}. Strategic
combinatorial bandit mechanisms build on the broader combinatorial bandit
framework~\cite{chen2013combinatorial,gao2021auction}.

\subsection{A Sufficient Coverage Condition}
If the contexts are i.i.d. with \(\Prb(c_t=c)\ge\pi_{\min}>0\) for every \(c\),
and exploration selects each producer uniformly and independently of bids, then
Chernoff bounds give Assumption~\ref{ass:procurement-coverage} with
\(c_{\mathrm{cov}}=g\pi_{\min}/2\), provided
\(M\ge C_0(n/\pi_{\min})\log(Tng)\) for a sufficiently large universal
constant \(C_0\).

\subsection{Scope of F-UCB Incentives}
\label{app:scope-fucb-incentives}
Frozen-payment UCB is not claimed to be ex-post individually rational: its
adaptive allocation rule and frozen payments need not agree round by round.
The withdrawal action in Theorem~\ref{thm:ucb-post-exploration} nevertheless
gives approximate individual rationality under the same benchmark. Writing
\(\mathcal U_{i,\mathrm{post}}^{\FUCB}\) for post-exploration utility, it gives
\[
\E_{H_M}\!\left[
\left(
-\E[\mathcal U_{i,\mathrm{post}}^{\FUCB}(k_i,b_{-i};k_i)\mid H_M]
\right)_+
\right]
\le
\epsilon_{i,T}^{\mathrm{post}}(k_i;b_{-i}).
\]

\section{Balanced Tuning}
\label{app:balanced-tuning}

The main text highlights the near-UCB tuning, which preserves the
\(\wt O(\sqrt{ngT})\) welfare-regret order of UCB. The alternative balanced
tuning chooses
\[
    M_{\mathrm{bal}}
    \asymp
    (ng)^{1/3}T^{2/3},
\]
capped at \(T\), and uses the margin scale
\[
    \rho_{T,\mathrm{bal}}^\star
    =
    C_{\mathrm{gap}}\sqrt{ng\log(2Tng)}\,T^{-1/3}.
\]
This choice makes the truthful-path mismatch term no larger than the frozen
payment-estimation term. In the smoothed model of
Remark~\ref{rem:when-margin-holds}, the exceptional probability is
\(O(Bg(n+1)^2\rho_{T,\mathrm{bal}}^\star)\).

\begin{proposition}[Balanced F-UCB regret]
\label{prop:balanced-fucb-regret}
Under the hypotheses of Theorem~\ref{thm:fucb-regret}, F-UCB with
\(M=M_{\mathrm{bal}}\) satisfies
\[
    \E[\Reg(T)]
    =
    \wt O\!\left((ng)^{1/3}T^{2/3}+\sqrt{ngT}\right).
\]
In particular, for fixed \(n\) and \(g\), this is \(\wt O(T^{2/3})\).
\end{proposition}

\begin{proof}
Substitute \(M\asymp(ng)^{1/3}T^{2/3}\) into the bound
\(\wt O(M+\sqrt{ngT})\) from Theorem~\ref{thm:fucb-regret}. For fixed
\(n,g\), the \(T^{2/3}\) term dominates the \(\sqrt T\) term.
\end{proof}

\begin{proposition}[Balanced approximate truthfulness]
\label{prop:balanced-approx-truth}
Under the hypotheses of Theorem~\ref{thm:ucb-direct-approx-truth}, if
\(M=M_{\mathrm{bal}}\) and
\(\rho_T\ge \rho_{T,\mathrm{bal}}^\star\), then F-UCB is
\[
    \epsilon_T
    =
    \wt O\!\left((ng)^{1/3}T^{2/3}\right)
\]
approximately truthful. For fixed \(n,g\), the cumulative deviation gain is
\(\wt O(T^{2/3})\), and the average per-round incentive is
\(\wt O(T^{-1/3})\). The same bound holds for unilateral deviations against
truthful opponents by taking \(b_{-i}=k_{-i}\).
\end{proposition}

\begin{proof}
Theorem~\ref{thm:ucb-direct-approx-truth} gives
\[
\epsilon_T
=
\wt O\!\left(
\frac{ng}{\rho_T^2}
+
T\sqrt{\frac{ng}{M}}
\right).
\]
With \(M\asymp(ng)^{1/3}T^{2/3}\), the payment-estimation term is
\(\wt O((ng)^{1/3}T^{2/3})\). With
\(\rho_T\ge\rho_{T,\mathrm{bal}}^\star\), the truthful-path mismatch term is
\(\wt O(T^{2/3})\), which is no larger up to the displayed \(n,g\) dependence.
Dividing by \(T\) gives the average per-round rate for fixed \(n,g\). The final
claim is Corollary~\ref{cor:fucb-truthful-opponents}.
\end{proof}

\begin{corollary}[Balanced smoothed approximate truthfulness]
\label{cor:balanced-smoothed-approx-truth}
Fix producer \(i\) and other bids \(b_{-i}\), and consider the truthful-path
profile \(b^0=(k_i,b_{-i})\) generated by a smoothed instance satisfying the
density condition in Remark~\ref{rem:when-margin-holds} with bound \(B\). Under
Assumption~\ref{ass:payment-regularity}, F-UCB with the balanced tuning is
\(\wt O(T^{2/3})\)-approximately truthful, with average per-round incentive
\(\wt O(T^{-1/3})\), at \((i,k_i,b_{-i})\) for fixed \(n,g\), with probability at
least
\[
    1-O\!\left(Bg(n+1)^2\rho_{T,\mathrm{bal}}^\star\right).
\]
The probability is over the smoothed instance; conditional on the realized
instance, the utility expectations are over contexts, reward noise, and
mechanism randomness.
\end{corollary}

\begin{proof}
Apply Remark~\ref{rem:when-margin-holds} with
\(\rho=\rho_{T,\mathrm{bal}}^\star\). Conditional on the margin event,
Proposition~\ref{prop:balanced-approx-truth} applies.
\end{proof}

\begin{proposition}[Balanced post-exploration benchmark]
\label{prop:balanced-post-exploration}
Under the hypotheses of Theorem~\ref{thm:ucb-post-exploration}, if
\(M=M_{\mathrm{bal}}\) and
\(\rho_T\ge \rho_{T,\mathrm{bal}}^\star\), then
\[
    \epsilon_{i,T}^{\mathrm{post}}(k_i;b_{-i})
    =
    \wt O\!\left((ng)^{1/3}T^{2/3}\right).
\]
For fixed \(n,g\), this gives cumulative post-exploration deviation gain
\(\wt O(T^{2/3})\) and average gain \(\wt O(T^{-1/3})\).
\end{proposition}

\begin{proof}
Substitute the balanced values of \(M\) and \(\rho_T\) into
Theorem~\ref{thm:ucb-post-exploration}, exactly as in the proof of
Proposition~\ref{prop:balanced-approx-truth}.
\end{proof}

\section{Experimental Details}
\label{app:experimental-details}

Both experiments use the fixed instance
\[
\begin{aligned}
    \mu =
    \begin{pmatrix}
        0.55 & 0.70 & 0.95\\
        0.90 & 0.39 & 0.56\\
        0.46 & 0.92 & 0.61
    \end{pmatrix},
    \quad
    k=(0.25,0.17,0.22),\\
    \Pr(c)=(0.31,0.37,0.32).
\end{aligned}
\]
Contexts are sampled independently from this distribution. ETC uses
\(M=(ng)^{1/3}T^{2/3}\), rounded to the nearest integer, bid-independent
exploration rounds. F-UCB uses \(M=T^{1/2}\), also rounded to the nearest
integer, bid-independent exploration rounds for its frozen payments; this is the
fixed-$n,g$ specialization of the near-UCB tuning. It then updates allocation
estimates using UCB. During
exploration, the producer is selected uniformly conditional on the realized
context. Individual exploration rewards are Bernoulli with mean
$\mu_{i,c}$. To batch the large-horizon post-exploration computation, subsequent
feedback is set equal to its conditional mean $\mu_{i,c}$. This remains a bounded
conditional-mean reward process of the form assumed in the model, while removing
only post-exploration sampling variance.

The implementation uses
\[
\beta_{i,c}(t)
=
\sqrt{
\frac{2\log(\max\{4,2Tng\})}
{N_{i,c}(t)\vee1}
},
\]
corresponding to $C=2$ in the theoretical UCB rule. Ties between producers are
resolved in favor of the smallest index, and the outside option is selected
whenever the largest producer score is nonpositive. These rules are fixed ex
ante and bid-independent. Figure~\ref{fig:exp-truthful-regret} uses 120
repetitions for $T\le10^5$, 60 for $10^5<T\le10^6$, and 20 for $T>10^6$; its
bands show one standard error. Table~\ref{tab:epsnet-equilibrium} uses 40
repetitions at every reported horizon.

For the strategic-bidding experiment, producers choose bids from the grid
\(\{0,0.1,0.2,\ldots,1\}\), augmented with the truthful costs. Thus
\(\epsilon=0.1\) is the base bid-grid spacing, not an equilibrium-approximation
parameter. For each simulated instance, a bid-grid search initialized at the
truthful profile returns the strategic profile at which welfare regret is
evaluated; implementation details are provided with the accompanying code. This
is an ex-post stress test: the profile may depend on the realized simulation
instance, whereas producers in the mechanism submit bids before the realized
reward and context sample paths are known. The ETC column in
Table~\ref{tab:epsnet-equilibrium} reports truthful ETC regret. It is included
as a benchmark rather than as a bid-grid search result.

\section{Deferred Proofs}
\label{app:proofs}

\begin{proof}[Proof of Proposition~\ref{prop:oracle-dsic}]
Fix $i$, a context $c$, and $b_{-i}$, and write
$\theta_i(c;b_{-i})=\mu_{i,c}-H_i(c;b_{-i})$.
Producer $i$ can win only by reporting a bid no larger than
$\theta_i(c;b_{-i})$, with endpoint ties resolved by the fixed tie-breaking
rule. If $\theta_i(c;b_{-i})<0$, no feasible bid in $[0,1]$ wins and truthful
utility is zero. If $\theta_i(c;b_{-i})>1$, every feasible bid wins and the
payment is $1$, so truthful utility is $1-k_i\ge0$ and no report changes the
allocation. Finally, if $\theta_i(c;b_{-i})\in[0,1]$, the usual critical-value
argument applies with critical ask $q_i^\star(c;b_{-i})=\theta_i(c;b_{-i})$:
types below the threshold prefer winning at the threshold payment, types above
it prefer losing, and equality gives zero utility either way. These three cases
prove dominant-strategy truthfulness and ex-post individual rationality.
\end{proof}

\begin{proof}[Proof of Proposition~\ref{thm:ucb-procurement-regret}]
For each pair $(i,c)$, enumerate the observations collected when product $i$ is
selected in context $c$. By the bounded martingale version of Hoeffding's
inequality and a union bound over all $(i,c)$ pairs and all sample counts up
to $T$, with probability at least $1-(Tng)^{-2}$, the event
$\mathcal E=\{\forall i,c,t:
|\wh\mu_{i,c}(t)-\mu_{i,c}|\le\beta_{i,c}(t)\}$ holds for \(C\) sufficiently
large in Hoeffding's bound.

Condition on $\mathcal E$, and let
$i_t^\star\in\arg\max_{a\in[n]\cup\{\empt\}}S_a(c_t;k)$ be the
full-information procurement choice under truthful bidding. If \(N_{i,c_t}(t)=0\),
then \(\beta_{i,c_t}(t)\ge1\) and \(\wh\mu_{i,c_t}(t)=0\), so
\(\wt\mu_{i,c_t}(t)\ge\mu_{i,c_t}\). If \(N_{i,c_t}(t)>0\), the same inequality
follows from \(\mathcal E\). Thus every optimistic score is at least its true
score.

Because $i_t$ maximizes the optimistic score,
\[
    \wt S_{i_t}(c_t;k,t)
    \ge
    \wt S_{i_t^\star}(c_t;k,t)
    \ge
    S_{i_t^\star}(c_t;k).
\]
If $i_t\neq\empt$, then
\[
\begin{aligned}
    S_{i_t^\star}(c_t;k)-S_{i_t}(c_t;k)
    &\le
    \wt S_{i_t}(c_t;k,t)-S_{i_t}(c_t;k)  \\
    &=
    \wt\mu_{i_t,c_t}(t)-\mu_{i_t,c_t}     \\
    &\le
    2\beta_{i_t,c_t}(t).
\end{aligned}
\]
If $i_t=\empt$, then every optimistic producer score is nonpositive. Since true
scores are no larger than optimistic scores on $\mathcal E$, all true producer
scores are also nonpositive, so the outside option is full-information optimal and
the regret in that round is zero.

Therefore, on $\mathcal E$,
\[
    \Reg(T)
    \le
    2\sum_{t:i_t\neq\empt}\beta_{i_t,c_t}(t).
\]
The standard counting argument (Cauchy--Schwarz over the $ng$ pairs and their
sample counts) yields
\[
    \sum_{t:i_t\neq\empt}\beta_{i_t,c_t}(t)
    =
    \wt O(\sqrt{ngT}).
\]
On the failure event, the per-round regret is at most $O(1)$, so its
contribution to expected regret is
$O(T\cdot(Tng)^{-2})=O(1/(Tn^2g^2))$. This proves the proposition.
\end{proof}

\begin{proof}[Proof of Theorem~\ref{thm:etc-procurement-regret}]
The exploration phase lasts $M$ rounds. Since each true score lies in
$[-1,1]$, the regret in any single round is at most $2$, so the exploration
regret is at most $O(M)$.

By Assumption~\ref{ass:procurement-coverage} and the bounded martingale
concentration inequality, with probability at least $1-2(Tng)^{-2}$,
$\varepsilon_M:=\max_{i,c}|\wh\mu^0_{i,c}-\mu_{i,c}|
=\wt O(\sqrt{ng/M})$.
Condition on this event. For any context $c$, let
$i^\star(c)\in\arg\max_{a\in[n]\cup\{\empt\}}S_a(c;k)$ be the
full-information choice and
$\wh i(c)\in\arg\max_{a\in[n]\cup\{\empt\}}\wh S_a^0(c;k)$ the empirical
choice.

Since every empirical score differs from the corresponding true score by at
most $\varepsilon_M$, the standard plug-in optimality argument gives
$S_{i^\star(c)}(c;k)-S_{\wh i(c)}(c;k)\le2\varepsilon_M$, where the outside
option has score $0$. Therefore the commit-phase regret is at most
$O(T\varepsilon_M)=\wt O(T\sqrt{ng/M})$.
On the complementary event, the total regret is at most $O(T)$, whose
expected contribution is negligible. Adding the exploration and commit terms
gives the stated regret bound. The displayed choices of $M$ give the
corresponding rates.
\end{proof}

\begin{proof}[Proof of Theorem~\ref{thm:etc-procurement-truthfulness}]
During exploration, neither allocation nor payments depend on bids. Therefore a
producer cannot affect its exploration allocation, its exploration payment, or
the data used to construct $\wh\mu^0$ by changing its bid.

Now condition on any realized exploration history and hence on a fixed estimate
$\wh\mu^0$. In the commit phase, producer $i$'s score in context $c$ is
$\wh\mu^0_{i,c}-b_i$.
Since this score is strictly decreasing in $b_i$ and tie-breaking is fixed ex
ante, the allocation rule is monotone nonincreasing in $b_i$. The payment is
the clipped critical ask obtained from the fixed estimates $\wh\mu^0$. Thus the
same three-case threshold argument as in Proposition~\ref{prop:oracle-dsic}
proves truthfulness and ex-post individual rationality conditional on this
exploration history.

Since the exploration phase is bid-independent and the commit phase is truthful
conditional on every exploration history, truthful bidding is a dominant
strategy for the entire mechanism. The exploration payment $1$ covers every
cost in $[0,1]$, and the commit-phase critical payment is ex-post individually
rational by Proposition~\ref{prop:oracle-dsic}, proving the final claim.
\end{proof}

\begin{proof}[Proof of Theorem~\ref{thm:fucb-regret}]
The exploration phase contributes at most $2M=O(M)$ regret. Afterwards,
the allocation rule is precisely the UCB rule of
Proposition~\ref{thm:ucb-procurement-regret}, initialized with the exploration
observations. Initial observations can only decrease the subsequent
confidence radii. Repeating the counting argument in that proposition over the
$T-M$ post-exploration rounds gives $\wt O(\sqrt{ngT})$ expected regret,
including the negligible concentration-failure contribution.
\end{proof}

\begin{proof}[Proof of Lemma~\ref{lem:frozen-critical-payment-accuracy}]
By uniform exploration coverage and the bounded martingale concentration
inequality, with probability at
least $1-2(Tng)^{-2}$, $\varepsilon_M=\wt O(\sqrt{ng/M})$.
Condition on this event.

Fix producer $i$, context $c$, and other bids $b_{-i}$. Define the true
threshold faced by producer $i$ as
$H_i(c;b_{-i})=\max\{0,\max_{j\neq i}(\mu_{j,c}-b_j)\}$, and the frozen
empirical threshold as
$\wh H_i(c;b_{-i})=\max\{0,\max_{j\neq i}(\wh\mu^{\mathrm{pay}}_{j,c}-b_j)\}$.
Since every value estimate differs from its true value by at most
$\varepsilon_M$, and the $\max$ operator is $1$-Lipschitz, we have
$|\wh H_i(c;b_{-i})-H_i(c;b_{-i})|\le\varepsilon_M$.

By definition of the clipped VCG threshold,
$q_i^\star(c;b_{-i})=[\mu_{i,c}-H_i(c;b_{-i})]_{[0,1]}$, while
$\wh q_i(c;b_{-i})
=[\wh\mu^{\mathrm{pay}}_{i,c}-\wh H_i(c;b_{-i})]_{[0,1]}$.
The two scalar inputs differ by at most $2\varepsilon_M$. Since projection onto
$[0,1]$ is $1$-Lipschitz, $|\wh q_i(c;b_{-i})-q_i^\star(c;b_{-i})|
\le2\varepsilon_M$.
This proves the claim.
\end{proof}

\begin{proof}[Proof of Lemma~\ref{lem:truthful-ucb-stability}]
The proof is pathwise after fixing the truthful bid profile $b^0$. All empirical
means, sample counts, and confidence radii are those generated by the UCB
trajectory under $b^0$.

Condition on $\mathcal E$. In a post-exploration round $t$ with context
$c_t=c$, let $\wt S_j(c;b^0,t)=\wh\mu_{j,c}(t)+\beta_{j,c}(t)-b^0_j$
for $j\in[n]$, and let $\wt S_{\empt}(c;b^0,t)=0$. On $\mathcal E$,
$S_j(c;b^0)\le\wt S_j(c;b^0,t)\le S_j(c;b^0)+2\beta_{j,c}(t)$ for every
producer $j$.

Let $a^\star=a^\star(c;b^0)$. First suppose $a^\star\in[n]$. By the margin
condition, since the outside option is not optimal,
$S_{a^\star}(c;b^0)\ge\rho_T$.
Thus the optimistic score of $a^\star$ is positive, so UCB cannot choose the
outside option. If UCB selects an incorrect producer $j\neq a^\star$, then
$\wt S_j(c;b^0,t)\ge\wt S_{a^\star}(c;b^0,t)\ge S_{a^\star}(c;b^0)$, whereas
$\wt S_j(c;b^0,t)\le S_j(c;b^0)+2\beta_{j,c}(t)$. Hence
$2\beta_{j,c}(t)\ge S_{a^\star}(c;b^0)-S_j(c;b^0)\ge\rho_T$.

Now suppose $a^\star=\empt$. Then $S_j(c;b^0)\le-\rho_T$ for every $j\in[n]$.
If UCB selects producer $j$, then $\wt S_j(c;b^0,t)\ge0$, whereas
$\wt S_j(c;b^0,t)\le S_j(c;b^0)+2\beta_{j,c}(t)
\le-\rho_T+2\beta_{j,c}(t)$. Thus again
$\beta_{j,c}(t)\ge\rho_T/2$.

Thus every truthful-path allocation mistake can be charged to a selected
producer-context pair $(j,c)$ with $\beta_{j,c}(t)\ge\rho_T/2$. Since
$\beta_{j,c}(t)=\sqrt{C\log(2Tng)/(N_{j,c}(t)\vee1)}$, this can happen only
while $N_{j,c}(t)\le O(\log(2Tng)/\rho_T^2)$.
Each selection of $j$ in context $c$ increments $N_{j,c}$ by one, so the number
of such selections is at most this same quantity. Summing over all $ng$
producer-context pairs gives the claimed bound.
\end{proof}

\begin{proof}[Proof of Lemma~\ref{lem:report-uniform-comparison}]
Put $d_i^\star(c)=q_i^\star(c;b_{-i})-k_i$ and
$\wh d_i(c)=\wh q_i(c;b_{-i})-k_i$. In a run using any report,
let $x_t'$ be producer $i$'s selection indicator. Since $x_t'\in\{0,1\}$,
for every post-exploration context $c_t$,
\[
\wh d_i(c_t)x_t'\le \wh d_i(c_t)^+
\le d_i^\star(c_t)^++\eta(H_M).
\]
For the truthful run, let $x_t^0$ be its selection indicator and let
$x_t^\star$ indicate that the full-information truthful allocation selects
$i$. The clipped critical payment implies the two facts needed below: if
$x_t^\star=1$, then $d_i^\star(c_t)\ge0$; if $x_t^\star=0$, then
$(d_i^\star(c_t))^+=0$. Boundary cases created by clipping or tie-breaking have
zero utility and satisfy the same inequalities. Hence
\[
\wh d_i(c_t)x_t^0
\ge d_i^\star(c_t)^+-\eta(H_M)-\1\{x_t^0\neq x_t^\star\}.
\]
If $x_t^0=x_t^\star$, this follows from the payment error bound. Indeed, when
both indicators are one, $\wh d_i(c_t)\ge d_i^\star(c_t)-\eta(H_M)$; when both
are zero, the right-hand side is at most \(0\) because
\((d_i^\star(c_t))^+=0\). If the indicators differ, there are two cases. If
$x_t^0=0$ and $x_t^\star=1$, then the right-hand side is at most
$d_i^\star(c_t)-1\le0$, since the clipped payment and \(k_i\in[0,1]\) imply
\(d_i^\star(c_t)\le1\). If
$x_t^0=1$ and $x_t^\star=0$, then \((d_i^\star(c_t))^+=0\), so the right-hand
side is at most \(-1\), while \(\wh d_i(c_t)\ge-1\).

The two reports induce the same conditional law of future
contexts, because contexts are exogenous and action-independent. Taking
conditional expectations in the first display under the deviating report and
in the second display under truthful bidding, then subtracting, bounds the
expected gain by
$2(T-M)\eta(H_M)$ plus the expected number of post-exploration rounds with
$x_t^0\neq x_t^\star$. This selection-indicator mismatch is at most the action
mismatch that defines $Z(H_M)$, so the claim follows. No coupling of the two
adaptive reward paths is needed.
\end{proof}

\begin{proof}[Proof of Theorem~\ref{thm:ucb-direct-approx-truth}]
Lemma~\ref{lem:report-uniform-comparison} gives the same upper bound for the
post-exploration supremum; in particular, the expected gain of any fixed report
is at most $\E[2(T-M)\eta(H_M)+Z(H_M)]$. The
payment-accuracy lemma gives
$\eta(H_M)=\wt O(\sqrt{ng/M})$ on its high-probability event, while always
$\eta(H_M)\le1$ because both payments are clipped to $[0,1]$. Hence
$\E[\eta(H_M)]=\wt O(\sqrt{ng/M})$.
For the truthful UCB path, Lemma~\ref{lem:truthful-ucb-stability} bounds the
realized number of allocation mistakes by $\wt O(ng/\rho_T^2)$ on the UCB
concentration event. On the complement, the number of mistakes is at most $T$,
and the concentration failure probability is $O((Tng)^{-2})$. Thus
$\E[Z(H_M)]=\wt O(ng/\rho_T^2)$. This proves both conclusions.
\end{proof}

\begin{proof}[Proof of Corollary~\ref{cor:smoothed-approx-truth}]
Apply Remark~\ref{rem:when-margin-holds} with
\(\rho=\rho_{T,\mathrm{near}}^\star\). Conditional on the corresponding margin
event, Theorem~\ref{thm:ucb-direct-approx-truth} applies. Substituting the
near-UCB tuning gives the displayed rate after suppressing logarithmic factors.
\end{proof}

\begin{proof}[Proof of Corollary~\ref{cor:fucb-truthful-opponents}]
Apply Theorem~\ref{thm:ucb-direct-approx-truth} with $b_{-i}=k_{-i}$.
\end{proof}

\begin{proof}[Proof of Theorem~\ref{thm:ucb-post-exploration}]
Lemma~\ref{lem:report-uniform-comparison} bounds the post-exploration supremum,
including the withdrawal action, by $2(T-M)\eta(H_M)+Z(H_M)$. The estimates
in the preceding proof give $\E[\eta(H_M)]=\wt O(\sqrt{ng/M})$ and
$\E[Z(H_M)]=\wt O(ng/\rho_T^2)$. Taking expectations proves the theorem.
\end{proof}

\begin{proof}[Proof of Lemma~\ref{lem:frozen-critical-price-lb}]
It suffices to prove the lower bound even if the mechanism is given \(M\) direct
Bernoulli samples from the producer, since this only gives the mechanism more
information than an arbitrary bid-independent payment-learning history with at
most \(M\) such samples.

Let \(\mu_+=1/2+\Delta\) and \(\mu_-=1/2-\Delta\).
For \(c_0\) sufficiently small and \(M\) sufficiently large, both values lie in
\([1/3,2/3]\). The corresponding critical asks are
\(q^\star_+=\mu_+=1/2+\Delta\) and
\(q^\star_-=\mu_-=1/2-\Delta\), so
\(|q^\star_+-q^\star_-|=2\Delta\).

For Bernoulli distributions bounded away from \(0\) and \(1\), there is a
universal constant \(C\) such that
\[
    \mathrm{KL}\!\left(
        \mathrm{Bern}(\mu_+) \,\middle\|\, \mathrm{Bern}(\mu_-)
    \right)
    \le
    C(\mu_+-\mu_-)^2.
\]
Therefore
\[
    \mathrm{KL}(P_+^M\,\|\,P_-^M)
    \le
    C M(\mu_+-\mu_-)^2
    =
    4 C M\Delta^2
    =
    4 C c_0^2.
\]
Choosing \(c_0\) small enough and applying Pinsker's inequality gives
\(\mathrm{TV}(P_+^M,P_-^M)\le1/4\).

No estimator \(\widehat q\) can be within distance \(\Delta/2\) of both
\(q^\star_+\) and \(q^\star_-\). Hence
\[
\begin{aligned}
    &P_+\!\left(
        |\widehat q-q^\star_+|\ge \frac{\Delta}{2}
    \right)
    +
    P_-\!\left(
        |\widehat q-q^\star_-|\ge \frac{\Delta}{2}
    \right)
    \\
    &\qquad\ge
    1-\mathrm{TV}(P_+^M,P_-^M)
    \ge
    \frac34 .
\end{aligned}
\]
Therefore at least one of the two values
\(\mu\in\{\mu_+,\mu_-\}\) satisfies
\[
    P_\mu\!\left(
        |\widehat q-q^\star|\ge \frac{\Delta}{2}
    \right)
    \ge
    \frac38.
\]
On this event, \(A(H_M)\ge\Delta/4\), and hence
\(\mathbb E_\mu[A(H_M)]\ge(3/8)(\Delta/4)=3\Delta/32\).
The claim follows with \(c_1=3/32\). The case of small \(M\) can be absorbed by
adjusting constants.
\end{proof}

\begin{proof}[Proof of Theorem~\ref{thm:frozen-payment-lb}]
Fix the value \(\mu\) supplied by the previous lemma, and write
\(q^\star=\mu\), \(\Delta=\Theta(M^{-1/2})\),
\(k^-=q^\star-\Delta/4\), and \(k^+=q^\star+\Delta/4\).
Let \(L=T-M\). Conditional on \(H_M\), for a reported bid \(b\), define
\[
    X_b(H_M)
    :=
    \mathbb E
    \left[
        \sum_{t=M+1}^T \mathbf 1\{i_t(b)=1\}
        \,\middle|\, H_M
    \right],
\]
where the expectation is over post-learning reward noise and internal
randomness. The true cost affects utility but not the reward distribution, so
\(X_b(H_M)\) depends on \(b\), not on the true cost. For a true cost \(k\), let
\(U(b;k)\) denote the random total payment-minus-cost utility accrued over the
\(L\) post-learning rounds when the producer reports \(b\), and set
\[
\Gamma_b(k;H_M):=\mathbb E[U(b;k)-U(k;k)\mid H_M],
\]
\[
    \epsilon_T^{\mathrm{post}}(k)
    :=
    \mathbb E_{H_M}
    \left[
        \sup_{b'\in[0,1]\cup\{\bot\}}
        \left(\Gamma_{b'}(k;H_M)\right)_+
    \right],
\]
where \(U(\bot;k)=0\).

Under cost \(k^-\), the producer is full-information optimal because
\(\mu-k^-=\mu-q^\star+\Delta/4=\Delta/4>0\).
Thus each post-learning round in which the producer is not selected incurs
regret \(\Delta/4\).

Under cost \(k^+\), the outside option is full-information optimal because
\(\mu-k^+=\mu-q^\star-\Delta/4=-\Delta/4<0\).
Thus each post-learning round in which the producer is selected incurs
regret \(\Delta/4\).

For a realized payment-learning history \(H_M\), write
\(X_-=X_{k^-}(H_M)\) and \(X_+=X_{k^+}(H_M)\).
The conditional post-learning regret under truthful bidding at cost \(k^-\)
is \(\Delta(L-X_-)/4\),
and the conditional post-learning regret under truthful bidding at cost
\(k^+\) is \(\Delta X_+/4\). Therefore
\(\mathbb E[L-X_-]=4R_T^{\mathrm{post}}(k^-)/\Delta\) and
\(\mathbb E[X_+]=4R_T^{\mathrm{post}}(k^+)/\Delta\).

Now we lower bound profitable deviations. Let
\(d(H_M):=X_- - X_+\).

First consider the low-cost type \(k^-\). Its truthful and upward-deviation
utilities are respectively \((\widehat q-k^-)X_-\) and
\((\widehat q-k^-)X_+\).
Thus define the positive gain from this upward deviation as
\[
    G_-(H_M)
    :=
    \left(
        (\widehat q-k^-)(X_+-X_-)
    \right)_+ .
\]

Next consider the high-cost type \(k^+\). Its truthful and downward-deviation
utilities are respectively \((\widehat q-k^+)X_+\) and
\((\widehat q-k^+)X_-\).
Thus define the positive gain from this downward deviation as
\[
    G_+(H_M)
    :=
    \left(
        (\widehat q-k^+)(X_- - X_+)
    \right)_+ .
\]

Let \(G(H_M):=G_-(H_M)+G_+(H_M)\). We claim that
\[
\begin{aligned}
    G(H_M)
    &\ge
    \left[
        (k^- - \widehat q)_+
        +
        (\widehat q-k^+)_+
    \right]\\
    &\qquad\cdot d(H_M)_+ .
\end{aligned}
\]
Indeed, if \(d(H_M)\ge0\), then
\(G_-(H_M)=(k^- - \widehat q)_+d(H_M)\) and
\(G_+(H_M)=(\widehat q-k^+)_+d(H_M)\), so the inequality holds with
equality. If \(d(H_M)<0\), then
\(d(H_M)_+=0\), and the right-hand side is zero while
\(G(H_M)\ge0\).

Because \(k^-=q^\star-\Delta/4\) and \(k^+=q^\star+\Delta/4\), the
bracketed coefficient equals \((|\widehat q-q^\star|-\Delta/4)_+\).
Recall the truncated error
\[
    A(H_M)
    :=
    \min\left\{
        \left(
            |\widehat q-q^\star|-\frac{\Delta}{4}
        \right)_+,
        \frac{\Delta}{4}
    \right\}.
\]
Since \(A(H_M)\) is no larger than that coefficient and \(d(H_M)_+\ge0\),
we have \(G(H_M)\ge A(H_M)d(H_M)_+\ge A(H_M)d(H_M)\).
Using \(d(H_M)=X_- - X_+=L-(L-X_-)-X_+\), we get
\[
\begin{aligned}
    G(H_M)
    &\ge
    A(H_M)L
    -
    A(H_M)(L-X_-)
    -
    A(H_M)X_+.
\end{aligned}
\]
Since \(0\le A(H_M)\le\Delta/4\),
\[
\begin{aligned}
    G(H_M)
    &\ge
    A(H_M)L
    -
    \frac{\Delta}{4}(L-X_-)
    -
    \frac{\Delta}{4}X_+.
\end{aligned}
\]
Taking expectations over \(H_M\), we obtain
\[
\begin{aligned}
    \mathbb E[G]
    &\ge
    L\mathbb E[A(H_M)]
    -
    \frac{\Delta}{4}\mathbb E[L-X_-]
    -
    \frac{\Delta}{4}\mathbb E[X_+]
    \\
    &=
    L\mathbb E[A(H_M)]
    -
    R_T^{\mathrm{post}}(k^-)
    -
    R_T^{\mathrm{post}}(k^+).
\end{aligned}
\]
By the statistical lemma, \(\mathbb E[A(H_M)]\ge c_1\Delta\).
Therefore
\[
    \mathbb E[G]
    \ge
    c_1L\Delta
    -
    R_T^{\mathrm{post}}(k^-)
    -
    R_T^{\mathrm{post}}(k^+).
\]
Since \(L=T-M\) and \(\Delta=\Theta(M^{-1/2})\),
\[
    \mathbb E[G]
    \ge
    \Omega\!\left(
        \frac{T-M}{\sqrt M}
    \right)
    -
    R_T^{\mathrm{post}}(k^-)
    -
    R_T^{\mathrm{post}}(k^+).
\]

For each realized history \(H_M\), \(G_-(H_M)\) lower bounds the positive
gain available to type \(k^-\) from reporting \(k^+\). Similarly,
\(G_+(H_M)\) lower bounds the positive gain available to type
\(k^+\) from reporting \(k^-\). Thus, by the definition of the
post-learning incentive error,
\[
    \epsilon_T^{\mathrm{post}}(k^-)
    \ge
    \mathbb E[G_-],
    \qquad
    \epsilon_T^{\mathrm{post}}(k^+)
    \ge
    \mathbb E[G_+].
\]
Hence
\[
\begin{aligned}
    &\epsilon_T^{\mathrm{post}}(k^-)
    +
    \epsilon_T^{\mathrm{post}}(k^+)
    +
    R_T^{\mathrm{post}}(k^-)
    +
    R_T^{\mathrm{post}}(k^+)\\
    &\qquad\ge
    \Omega\!\left(
        \frac{T-M}{\sqrt M}
    \right).
\end{aligned}
\]

Since \(R_T^{\mathrm{post}}(k)\le R_T^{\mathrm{tot}}\) for every \(k\), and
the global \(\epsilon_T^{\mathrm{post}}\) is at least the incentive error of
either cost at the value of \(\mu\) fixed above,
we obtain, after changing universal constants, for some universal \(C>0\),
\[
    \epsilon_T^{\mathrm{post}}
    +
    C R_T^{\mathrm{tot}}
    \ge
    \Omega\!\left(
        \frac{T-M}{\sqrt M}
    \right).
\]
This proves the theorem.
\end{proof}

\begin{proof}[Proof of Lemma~\ref{lem:exploration-lb}]
Consider one context, one producer, and an outside option. Let the producer's
value be deterministic with \(\mu=1/2\). Compare \(k^0=0\) with \(k^1=1\).
Under \(k^0\), its score is \(\mu-k^0=1/2\),
so the producer is uniquely optimal. Selecting the outside option incurs regret
\(1/2\).

Under \(k^1\), its score is \(\mu-k^1=-1/2\),
so the outside option is uniquely optimal. Selecting the producer incurs regret
\(\frac12\).

In each of the \(M\) payment-learning rounds, the distribution of the allocation
decision used to collect the payment sample is independent of the submitted bid.
Since the reward distribution is also the same under \(k^0\) and \(k^1\), the
selection probability in such a round is the same under the two instances. Let
\(p\) be this conditional probability of selecting the producer. The expected
regret is \((1-p)/2\) under \(k^0\), and \(p/2\) under \(k^1\).
Averaging over the two instances gives
\(\frac12\cdot\frac12(1-p)+\frac12\cdot\frac12 p=1/4\). Hence the average
payment-learning regret per round over the two instances is at least \(1/4\).
Therefore the worst-case expected payment-learning regret is at least \(M/4\), and
so \(R_T^{\mathrm{tot}}\ge\Omega(M)\).
\end{proof}

\begin{proof}[Proof of Corollary~\ref{cor:frozen-payment-tradeoff-lb}]
The theorem gives constants \(c_1,C_1>0\) such that
\[
    \epsilon_T^{\mathrm{post}}
    +
    C_1R_T^{\mathrm{tot}}
    \ge
    c_1\frac{T-M}{\sqrt M}.
\]
The bid-independent payment-learning lower bound gives a constant \(c_2>0\) such
that \(R_T^{\mathrm{tot}}\ge c_2M\).
Choose \(C\ge C_1\). Then
\[
    \epsilon_T^{\mathrm{post}}
    +
    CR_T^{\mathrm{tot}}
    \ge
    c_1\frac{T-M}{\sqrt M},
\]
and also
\[
    \epsilon_T^{\mathrm{post}}
    +
    CR_T^{\mathrm{tot}}
    \ge
    Cc_2M.
\]
After changing constants and using \(\max\{a,b\}\ge(a+b)/2\), we obtain
\[
    \epsilon_T^{\mathrm{post}}
    +
    CR_T^{\mathrm{tot}}
    \ge
    c\left(
        M+\frac{T-M}{\sqrt M}
    \right).
\]
For \(M\le T/2\), this implies
$\epsilon_T^{\mathrm{post}}+CR_T^{\mathrm{tot}}
\ge c(M+T/\sqrt M)$ after another constant adjustment. This expression is
minimized at $M\asymp T^{2/3}$, where it is $\Omega(T^{2/3})$. Finally, if
$R_T^{\mathrm{tot}}=\wt O(T^{1/2})$, then
Lemma~\ref{lem:exploration-lb} forces $M=\wt O(T^{1/2})$. Hence
$T/\sqrt M=\wt\Omega(T^{3/4})$, and the regret term is lower order than
this quantity. Therefore
$\epsilon_T^{\mathrm{post}}=\wt\Omega(T^{3/4})$.
This proves the corollary.
\end{proof}

\end{document}